\documentclass[conference]{IEEEtran}

\IEEEoverridecommandlockouts
\usepackage{booktabs}
\usepackage{amssymb, cite}
\usepackage{bm,bbm}
\usepackage{mathtools}
\usepackage[size=normalsize]{caption}
\usepackage{multirow, color,amsfonts}
\usepackage{tabulary}
\usepackage{subfigure}
\usepackage{graphicx}
\usepackage{setspace}
\usepackage{enumerate}
\usepackage{comment}
\usepackage[utf8]{inputenc} 
\usepackage[T1]{fontenc}
\usepackage{ifthen}
\usepackage{cite}
\usepackage[hyphens]{url}
\usepackage[hidelinks]{hyperref}
\hypersetup{breaklinks=true}
\usepackage{amsthm}
\usepackage{amsmath}
\newtheorem{theorem}{Theorem}

\newtheorem{proposition}{Proposition}
\newtheorem{remark}{Remark}

\newtheorem{lemma}{Lemma}

\usepackage[dvipsnames]{xcolor}
\usepackage{graphicx} 
\usepackage{amssymb}
\usepackage{bbm}
\usepackage{algorithm}
\usepackage{algpseudocode}
\algrenewcommand{\algorithmicrequire}{\textbf{Input:}}
\algrenewcommand{\algorithmicensure}{\textbf{Output:}}

\begin{document}
\title{Learning-Augmented Perfectly Secure Collaborative Matrix Multiplication}

\author{
  \IEEEauthorblockN{Zixuan He\IEEEauthorrefmark{1}, M. Reza Deylam Salehi, Derya Malak\IEEEauthorrefmark{1}, Photios A. Stavrou\IEEEauthorrefmark{1}}
  \IEEEauthorblockA{\IEEEauthorrefmark{1}Communication Systems Department, EURECOM, Sophia-Antipolis, France\\
  \{\texttt{zixuan.he, reza.deylam-salehi, derya.malak,  fotios.stavrou}\}@eurecom.fr}
}

\maketitle

\begin{abstract}
This paper presents a perfectly secure matrix multiplication (PSMM) protocol for multiparty computation (MPC) of $\mathrm{A}^{\top}\mathrm{B}$ over finite fields. The proposed scheme guarantees correctness and information-theoretic privacy against threshold-bounded, semi-honest colluding agents, under explicit local storage constraints. Our scheme encodes submatrices as evaluations of sparse masking polynomials and combines coefficient alignment with Beaver-style randomness to ensure perfect secrecy. We demonstrate that any colluding set of parties below the security threshold observes uniformly random shares, and that the recovery threshold is optimal, matching existing information-theoretic limits. Building on this framework, we introduce a learning-augmented extension that integrates tensor-decomposition-based local block multiplication, capturing both classical and learned low-rank methods. We demonstrate that the proposed learning-based PSMM preserves privacy and recovery guarantees for MPC, while providing scalable computational efficiency gains (up to $80\%$) as the matrix dimensions grow.
\end{abstract}

\section{Introduction}
\label{sec:intro}
Matrix multiplication (MM) is a fundamental primitive in modern networked and distributed systems that underlies state estimation and control (e.g., Kalman filtering~\cite{simon:2006,stavrou:2018jstsp,stavrou:2024tac}, linear quadratic regulator~\cite{bertsekas:2005,stavrou:2021automatica,stavrou:2021tac2}), large-scale optimization~\cite{loshchilov2018large, smelyanskiy2005parallel}, and learning pipelines~\cite{strinati:2024,polyanskiy:2025isit}. In networked control systems (NCSs)~\cite{hespanha:2007,schlor:2021CDCmpc,darup:2021encrypted}, such computations increasingly involve multi-agent architectures~\cite{yuksel:2013} and shared computation infrastructures where multiple parties may contribute data, models, or computational resources while maintaining confidentiality requirements, which therefore makes security and privacy central concerns since semi-trusted agents may expose proprietary or safety-critical information. This motivates the investigation of exploiting \textit{information-theoretically secure} distributed MM protocols that (i) preserve perfect privacy against collusions, (ii) obey strict storage constraints, and (iii) achieve near-optimal recovery thresholds with minimal communication.

Classical secure multiparty computation (MPC) protocols such as the schemes of Shamir~\cite{Shamir:1979}, Blakley~\cite{blakley:1979}, and their extension to the BGW protocol~\cite{wigderson:1988BGW}, provide information-theoretic protection against collusion up to a threshold through polynomial-encoded functions. 
Subsequent triple-based and SPDZ-style systems~\cite{beaver:1991,keller:2020mpc, yang2025maliciously} improve the online phase but still retain high communication overhead for large-scale matrix operations. When directly applied to MM with per-worker storage constraints, these approaches typically require a larger number of agents to tolerate threshold collusion, resulting in excessive resource redundancy and processing latency.
Recent advances in coded computing~\cite{Yu:2017, Lee:2018tit, yu:2019pmlr, deyn0n2025, dutta2020, cartor2024secure, malak2024distributed, structured, deylam:2025graph} have shown that polynomial encodings can drastically reduce recovery thresholds in distributed linear algebra. 
In particular, \textit{polynomial sharing}~\cite{akbari:2021tit} achieves tight worker bounds for computing matrix polynomials under storage and resiliency constraints.
However, to the best of our knowledge, while existing frameworks may provide a unified coded MM primitive that guarantees \textit{perfect privacy}, \textit{bounded storage}, for distributed computation, they do \textit{not} incorporate intelligent or learning-based mechanisms such as low-rank decomposition or learned computation, to enhance local computational efficiency.

In this paper, we develop a perfectly secure matrix multiplication (PSMM) framework for multiparty computation, and establish guarantees for the correctness of reconstruction and information-theoretic privacy against collusion. Beyond security, we introduce, for the first time within this framework, a \textit{learning-augmented PSMM} (LA-PSMM) scheme that integrates data-driven low-rank structure into local computations to improve computational efficiency. Notably, this learning-based integration is shown to preserve the original information-theoretic privacy and recovery guarantees. The main contributions of this paper are listed as follows. 
\begin{itemize}
    \item We propose a finite-field multiparty protocol that achieves perfect privacy against threshold-bounded semi-honest collusions, while explicitly accounting for local storage constraints. The scheme encodes shared inputs as evaluations of sparse masking polynomials, where selected coefficients align with partitioned matrix blocks, and all remaining terms are masked using random coefficients inspired by Beaver’s triple construction (see eqs. \eqref{eq:polyA}, \eqref{eq:polyB}).
    \item We establish the correctness and privacy guarantees of the proposed protocol. In particular, we show that any colluding set of agents below the prescribed threshold observes polynomial evaluations that are uniformly random and statistically independent of the underlying secrets, thereby ensuring perfect information-theoretic privacy (see Lemma~\ref{lemma:masking}). Moreover, we prove that the number of agents required for exact recovery of all target blocks is sufficient and matches the optimum in coded computing~\cite{akbari:2021tit} (see Theorem~\ref{theorem:info-bound}, Proposition~\ref{prop:adv-dof}).
    \item We introduce a novel LA-PSMM framework that improves local computational efficiency without sacrificing security. In LA-PSMM, each agent replaces dense local block multiplication with a reduced-complexity tensor-decomposition approach. This framework couples modern learning-based decompositions (e.g., AlphaTensor~\cite{fawzi:2022matrixRL}) that are generalizations of classical approaches such as Strassen’s and Laderman's algorithm~\cite{strassen:1969,laderman:1976}. By exploiting learning-based decomposition, we enhance them through the integration of security and privacy, thereby enabling additional gains that benefit sparsity in local computation. We show that this operation preserves the polynomial support of local computations and thus does not affect the global recovery threshold, ensuring that privacy and correctness guarantees remain intact. Finally, numerical simulations validate LA-PSMM by demonstrating significant reductions in agent-side computational complexity compared with standard PSMM.
\end{itemize}




\textbf{Notation}: We denote the set of real numbers by $\mathbb{R}$. Different natural sets are defined as $\mathbb{N}\triangleq\{1,2,\ldots\}$, $\mathbb{N}_0\triangleq\{0,1,\ldots\}$, and $\mathbb{N}_j^N\triangleq\{j,\ldots,N\},~j\leq{N},~N\in\mathbb{N}$. 
We denote a finite field by $\mathbb{F}$ with $|\mathbb{F}|$ sufficiently large. Matrices are denoted by upright uppercase letters (e.g., $\mathrm{A}, \mathrm{B}$ and $\mathrm{A}^\top, \mathrm{B}^\top$ the transposes). For $k,m\in\mathbb{N}$ ($k<m$), we define the $k$ horizontal partition of a matrix $\mathrm{A}\in\mathbb{F}^{m\times m}$ by $\mathrm{A}=[\mathrm{A}_1\cdots\mathrm{A}_i\cdots\mathrm{A}_k]$, where $\mathrm{A}_i\in\mathbb{F}^{m\times(m/k)}$, $i\in\mathbb{N}_1^k$. We denote $\mathbf{0}$ a zero block of suitable size. For a matrix $\mathrm{A}\in\mathbb{F}^{m\times n}$, 
$\operatorname{vec}(\mathrm{A})\in\mathbb{F}^{m \cdot n)}$ 
denotes the column-wise vectorization of $\mathrm{A}$. The operator $\operatorname{mat}(\cdot)$ reshapes a vector in $\mathbb{F}^{m \cdot n)}$ into a matrix in $\mathbb{F}^{m\times n}$.

\section{Preliminaries}
\label{sec:pre}
We consider the multiparty computation setup as illustrated in Fig.~\ref{fig:ncs-mpc} in the NCS with a source plant, one controller (co-located with the actuator), and $N$ distributed computation agents, $N\in\mathbb{N}$. The agents are assumed to be connected by authenticated, private point-to-point channels with the controller and with the source plant, and to be semi-honest (passive adversary), with up to $t-1$ agents potentially colluding. The controller only learns the final computation outputs from the agents to apply its control law. 
\begin{figure}[h]
    \centering
    \includegraphics[width=\linewidth,height=3.5cm]{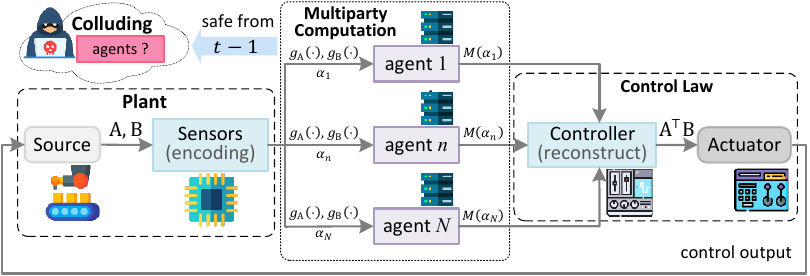}
    \caption{An NCS architecture incorporating a multiparty computation setup.}
    \label{fig:ncs-mpc}
\end{figure}

\textbf{Operation} (Protocol):
\textbf{(i)} \textit{Sharing}. The source plant generates two secret information matrices $\mathrm{A},\mathrm{B}\in\mathbb{F}^{m\times m}$ for some finite field $\mathbb{F}$ and dimension $m\in\mathbb{N}$, which need to be computed for multiplication aimed to be known by the controller. 
For example, these matrices $\mathrm{A},\mathrm{B}\in\mathbb{F}^{m\times m}$  can be seen as the correlation coefficient matrices of a multivariate time-invariant linear state-space model.
Due to the security concern, we use linear (additive) secret sharing, that is, Shamir sharing~\cite{Shamir:1979}, and Beaver's Triple~\cite{beaver:1991}. The secrets $\mathrm{A},\mathrm{B}$ will be partitioned into $k$ column blocks respectively
\begin{align}
    \mathrm{A}=[\mathrm{A}_1\cdots\mathrm{A}_j\cdots\mathrm{A}_k],~
    \mathrm{B}=[\mathrm{B}_1\cdots\mathrm{B}_j\cdots\mathrm{B}_k],
\label{eq:sharing-partition}
\end{align}
where $k\in\mathbb{N}, k<m$, and $\mathrm{A}_j,\mathrm{B}_j\in\mathbb{F}^{m\times(m/k)}$, $j\in\mathbb{N}_1^k$. Each partitioned sub-block of $\mathrm{A},\mathrm{B}$ is then encoded as the shares by leveraging a \textit{polynomial sharing} scheme similar to~\cite{akbari:2021tit}. This means that we encode the sub-blocks as coefficients of carefully designed sparse polynomials and distribute their evaluations as shares, which ensures that upon multiplying the encoded evaluations, the desired sub-block multiplications align in target coefficient positions, while random coefficients occupy all other positions to mask the secrets. Specifically, the sensor picks independent random masks so that each agent $n\in\mathbb{N}_1^N$ receives evaluations of two polynomials
\begin{align}
    g_\mathrm{A}(x) &= \sum_{j=1}^{k} \mathrm{A}_j x^{j-1} + \sum_{\ell=1}^{t-1} \mathrm{R}^{(\mathrm{A})}_\ell\, x^{k^2+\ell-1},
\label{eq:polyA}\\
    g_\mathrm{B}(x) &= \sum_{j=1}^{k} \mathrm{B}_j x^{k(j-1)} + \sum_{\ell=1}^{t-1} \mathrm{R}^{(\mathrm{B})}_l\, x^{k^2+\ell-1},
\label{eq:polyB}
\end{align}
at distinct public broadcasting points $\alpha_1,\ldots,\alpha_N\in\mathbb{F}$ to substitute $x$, where $\mathrm{R}^{(\mathrm{A})}_\ell,\mathrm{R}^{(\mathrm{B})}_\ell$ are i.i.d.\ uniform mask blocks \footnote{We view $\mathrm{R}^{(\mathrm{A})}_\ell,\mathrm{R}^{(\mathrm{B})}_\ell$ as $[a]$ and $[b]$ the components of a matrix-valued Beaver triple $([a],[b],[c]=[a]\times[b])$.} in $\mathbb{F}^{m\times (m/k)}$, and $g_\mathrm{A}, g_\mathrm{B}:\mathbb{F}^{m\times m}\rightarrow\mathbb{F}^{m\times m}$ denote polynomial functions. Hence, each agent $n\in\mathbb{N}_1^N$ stores  $g_\mathrm{A}(\alpha_n), g_\mathrm{B}(\alpha_n)$ that fit the storage budget. This placement mirrors the \textit{sparse} exponent pattern while ensuring coefficient alignment upon multiplication~\cite{akbari:2021tit}.\\
\noindent\textbf{(ii)} \textit{Local computation and occasional communication}. Each agent $n\in\mathbb{N}_1^N$ processes a local multiplication computation with the received $\alpha_n$ of the form
\begin{align}
    M(\alpha_n) = g_\mathrm{A}(\alpha_n)^{\top} g_\mathrm{B}(\alpha_n).
\label{eq:local-compute}
\end{align}
According to the design, by writing the multiplication computation in the summation form $M(x)=\sum_{\nu} \mathrm{M}_\nu x^\nu$, $\nu\in\mathbb{N}_0$, the coefficient $\mathrm{M}_\nu$ for any $i,j\in\mathbb{N}_1^k$ satisfies
\begin{align}
    \mathrm{M}_{i-1+k(j-1)} = \mathrm{A}_i^{\top}\mathrm{B}_j\, .
\label{eq:perfect-mul}
\end{align}
Therefore, all remaining coefficients include at least one random mask, which ensures privacy at the controller so that no exact information about $\mathrm{A}$ or $\mathrm{B}$ is carried. This coefficient-placement argument is standard in polynomial sharing~\cite{akbari:2021tit}. \\
\noindent \textbf{(iii)} \textit{Reconstruction of the result at the controller}. Each agent linearly recombines its available evaluation $M(\alpha_n)$ into polynomial shares of the target blocks, and transmits compact shares to the controller through the private channel. The controller then interpolates the needed coefficients to recover the target matrix. In particular, one multiplication requires only $N$ evaluations sufficient to interpolate all nonzero coefficients of $M(x)$. The degree and sparsity details are provided in the next section.

\textit{Objective}: The security goal of the considered system is to ensure \textit{correctness}, preserve \textit{privacy} against any $t-1$ colluding agents, and guarantee \textit{perfect information-theoretic privacy}, meaning that no additional information about the secrets is leaked to the controller beyond the agents’ outputs.

\section{Main Results}
\label{sec:main}
We present our main results in this section. We first give a lemma to show that privacy is maintained for the agent shares.

\begin{lemma}[Masking lemma]
\label{lemma:masking}
For any $p,q\in\mathbb{N}$, let $g(x)=\sum_{\ell=1}^{t-1} \mathrm{M}_\ell x^{\ell-1}$ with i.i.d.\ uniform coefficients $\mathrm{M}_\ell$ over $\mathbb{F}^{p\times q}$. For any distinct $\beta_1,\ldots,\beta_{t-1}\in\mathbb{F}$, the tuple $\big(g(\beta_1),\ldots,g(\beta_{t-1})\big)$ is uniformly distributed over $(\mathbb{F}^{p\times q})^{t-1}$ and independent of any other variables. 
\end{lemma}

\begin{proof}
\label{lemma-proof:masking}
By Lagrange interpolation, evaluation at $t-1$ distinct points is a bijection between coefficients and evaluations, and uniform coefficients imply uniform evaluations~\cite{Shamir:1979}. Hence, no information on other variables leaks through these masked evaluations.
\end{proof}

We next show the correctness maintained in the secret sharing in~\eqref{eq:polyA} and~\eqref{eq:polyB} during the agent's local computation. 
\begin{proposition}[Correctness of triple-based multiplication]
\label{prop:beaver}
For any $a,b,c\in\mathbb{N}$, let $\mathrm{A} \in \mathbb{F}^{a\times b}$ and $\mathrm{B}\in \mathbb{F}^{b\times c}$ be two secret matrices, and let $[\mathrm{A}],[\mathrm{B}]$ denote any two of their linear secret shares among the shared parties. 
Assume access to a Beaver triple $(\mathrm{R}_1,\mathrm{R}_2,\mathrm{R}_3)$ related to $[\mathrm{A}],[\mathrm{B}]$ with random matrices $\mathrm{R}_1\in \mathbb{F}^{a\times b}$, $\mathrm{R}_2\in \mathbb{F}^{b\times c}$ and $\mathrm{R}_3 = \mathrm{R}_1\mathrm{R}_2 \in \mathbb{F}^{a\times c}$. 
Define the \textit{opened} (i.e., publicly reconstructed) matrices
\begin{align}
    \mathrm{D} = \mathrm{A} - \mathrm{R}_1,~\text{and }~\mathrm{E}=\mathrm{B} - \mathrm{R}_2
\label{eq:beaver-open}
\end{align}
with the same respective dimensions. 
Then the recombined value $\mathrm{C}$ satisfies
\begin{align}
    \mathrm{C}
    = \mathrm{R}_3 + \mathrm{D}\mathrm{R}_2 + \mathrm{R}_1\mathrm{E} + \mathrm{D}\mathrm{E}
    = \mathrm{A}\mathrm{B} \in \mathbb{F}^{a\times c}.
\label{eq:beaver-compute}
\end{align}
Moreover, each term in $\mathrm{C}$ can be computed from local shares and the publicly opened $\mathrm{D,E}$ without revealing $\mathrm{A}$ or $\mathrm{B}$.
\end{proposition}

\begin{proof}
See Appendix~\ref{app:prop-proof:beaver}.
\end{proof}

The following result shows the information-theoretic privacy of the considered secret sharing scheme. 

\begin{theorem}[Agent bound for secure MM]
\label{theorem:info-bound}
Let $\mathrm{\mathrm{A,B}}\in \mathbb{F}^{m\times m}$ be $k$ partitioned respectively as
\begin{align}
    \mathrm{A}=\big[\mathrm{A}_1\cdots\mathrm{A}_i\cdots\mathrm{A}_k\big],~
    \mathrm{B}=\big[\mathrm{B}_1\cdots\mathrm{B}_j\cdots\mathrm{B}_k\big],
\label{eq:theo}
\end{align}
where $\mathrm{A}_i,\mathrm{B}_j \in \mathbb{F}^{m\times (m/k)}$, $i,j\in\mathbb{N}^k$. Assume each agent stores at most a $1/k$ fraction of each input (one $m\times (m/k)$ block per matrix), and up to $t-1$ agents may collude. There exists distinct encoding points $\alpha_1,\ldots,\alpha_N\in\mathbb{F}$ and encoders
\begin{align}
    g_\mathrm{A}(x) &= \sum_{i=1}^k \mathrm{A}_i\,x^{i-1} + \sum_{\ell=1}^{t-1} \mathrm{R}^{(\mathrm{A})}_\ell\,x^{k^2+\ell-1},
\label{eq:theo-polyA}\\
    g_\mathrm{B}(x) &= \sum_{j=1}^k \mathrm{B}_j\,x^{k(j-1)} + \sum_{\ell=1}^{t-1} \mathrm{R}^{(\mathrm{B})}_\ell\,x^{k^2+\ell-1}, 
\label{eq:theo-polyB}
\end{align}
with i.i.d.\ random masks $\mathrm{R}^{(\mathrm{A})}_\ell,\mathrm{R}^{(\mathrm{B})}_\ell\in  \mathbb{F}^{m\times (m/k)}$, such that the protocol that sends evaluations $g_\mathrm{A}(\alpha_n),g_\mathrm{B}(\alpha_n)$ to agent $n\in\mathbb{N}_1^N$ and aggregates $M(\alpha_n)  =  g_\mathrm{A}(\alpha_n)^\top g_\mathrm{B}(\alpha_n)$ to securely and perfectly reconstruct $\mathrm{A}^\top\mathrm{B}$ using $N$ agents, where
\begin{align}
   N \leq \min\big\{2k^2 + 2t - 3,\,  k^2 + kt + t - 2\,\big\}\ . 
\label{eq:info-bound}
\end{align}

\end{theorem}

\begin{proof}
See Appendix~\ref{app:theorem-proof:info-bound}.
\end{proof}

\begin{remark}[Positioning relative to polynomial sharing]
\label{rem-position-pol-shar}
The agent threshold (bound) in Theorem~\ref{theorem:info-bound} provides evaluations sufficient to interpolate all nonzero coefficients of $M(x)=g_\mathrm{A}(x)^\top g_\mathrm{B}(x)$ and hence obtain $\mathrm{A}^\top \mathrm{B}$ from every recovered multiplication $\mathrm{A}_i^\top\mathrm{B}_j$, $i,j\in\mathbb{N}^k$, and it coincides with the optimal recovery bounds known for general polynomial sharing schemes~\cite{akbari:2021tit}. However, our result is obtained under a more restrictive and structured model: we focus exclusively on MM, enforce a zero-gap alignment of all useful coefficients, and require closure under composition for subsequent multiplications. Achieving the same threshold under these additional constraints is nontrivial and essential for MPC-compatible and tensor-augmented secure computation.
\end{remark}

\begin{remark}[Improvement via polynomial sharing]
\label{rem-imp-soa}
The bound in Theorem~\ref{theorem:info-bound} strictly improves over naive job-splitting + BGW~\cite{ben:2019BGW}, which requires $k^2(2t-1)$ agents for one multiplication. Our scheme keeps $N$ independent of the number of chained multiplications when results are immediately reshared (\textit{closure} under polynomial sharing), as observed in~\cite{akbari:2021tit}.
\end{remark}

We next highlight the importance of our approach from the state-of-the-art through the following proposition.

\begin{proposition}[Advantage under reduced degrees-of-freedom]
\label{prop:adv-dof}
Consider the PSMM encoding of Theorem~\ref{theorem:info-bound} with storage fraction $1/k$ and privacy against $t-1$ colluding agents.
Let the target block multiplications be
\begin{align}
    \mathrm{Z}_{i,j} \triangleq \mathrm{A}_i^\top \mathrm{B}_j \in \mathbb{F}^{(m/k)\times(m/k)},~ i,j\in\mathbb{N}_1^k.
\label{eq:taget-block}
\end{align}
Assume the input class satisfies a known structural constraint such that the collection
\begin{align}
    \mathcal{Z} \triangleq \{\mathrm{Z}_{i,j}:~i,j\in\mathbb{N}_1^k\}
\label{eq:taget-blocks}
\end{align}
admits at most $s$ degrees-of-freedom (DOF), in the following sense:

\noindent\textbf{DOF assumption}: There exist known linear maps
$\mathcal{L}_1,\dots,\mathcal{L}_s$, for $s\in\mathbb{N}$, and known coefficient tensors
$\{\Gamma^{(i,j)}_\ell:l\in\mathbb{N}_1^s\}$ such that for all admissible inputs $(\mathrm{A,B})$,
\begin{equation}
    \mathrm{Z}_{i,j} \;=\; \sum_{\ell=1}^{s} \Gamma^{(i,j)}_\ell \, \mathcal{L}_\ell(\mathrm{A,B}), ~ \forall\, i,j\in\mathbb{N}_1^k.
\label{eq:dof-model}
\end{equation}
Equivalently, $\{\mathrm{Z}_{i,j}\}$ lies in a known $s$-dimensional linear subspace of
$\big(\mathbb{F}^{(m/k)\times(m/k)}\big)^{k^2}$.
Then the number of agent evaluations required to recover $\mathrm{A}^\top\mathrm{B}$ (and hence all needed $\mathrm{Z}_{i,j}$) with perfect privacy can be reduced from
\begin{align}
    N^\star(k,t)=\min\{\,2k^2+2t-3,\;k^2+kt+t-2\,\}\, ,
\end{align}
to
\begin{equation}
    N_{\text{struct}}(k,t,s) \;\leq\; \min\{\,2s+2t-3,\; s+ks+t-2\,\},
\label{eq:adv-newN}
\end{equation}
by (i) keeping the same PSMM masking/tail degree for $t$-privacy and
(ii) decoding only the $s$ effective DOF rather than all $k^2$ block coefficients.

In particular, for any fixed $N$, the tolerable collusion level satisfies
\begin{align}
    t_{\text{struct}} \;\ge\; t_{\text{worst}}\ ,
    \quad \forall s<k^2,
\label{eq:tolerable-collusion}
\end{align}
i.e., reduced DOF strictly increases collusion tolerance (or equivalently reduces the required number of agents).
\end{proposition}

\begin{proof}
    See Appendix~\ref{App:Proof-prop:adv-dof}.
\end{proof}

\section{Learning-Augmented Perfectly Secure Matrix Multiplication (LA-PSMM) Design}
\label{sec:learning-mul}
This section extends our PSMM protocol by incorporating a \textit{learning-based optimization layer} inspired by AlphaTensor~\cite{fawzi:2022matrixRL}. The goal is to jointly achieve algorithmic efficiency in the local multiplications via reinforcement learning (RL) while maintaining perfect information-theoretic security under the proposed PSMM framework. The resulting protocol, termed \textbf{Learning-Augmented PSMM (LA-PSMM)}, unifies our secret sharing and coded-computation guarantees with a data-driven discovery of efficient multiplication schemes.

\subsection{Motivation and Overview}
\label{subsec:learning-overview}
While PSMM ensures privacy and optimal agent thresholds, each agent is still required to perform a dense local MM
$M(\alpha_n) = g_\mathrm{A}(\alpha_n)^\top g_\mathrm{B}(\alpha_n)$, whose arithmetic complexity scales as $\mathcal{O}\big((\frac{m}{k})^3\big)$. Recent advances, such as AlphaTensor~\cite{fawzi:2022matrixRL}, demonstrate that deep RL can identify low-rank decompositions of the MM tensor, thereby reducing the number of scalar multiplications.

The LA-PSMM design equips each agent with an RL-discovered \textit{tensor decomposition pattern} that replaces the naive dense multiplication with a sequence of bilinear forms, while maintaining perfect secrecy through our masking and coded sharing mechanism (cf. Proposition~\ref{prop:agent-bound}).

\subsection{RL-Enhanced Block Multiplication}
\label{subsec:rl-mul}
Let each agent receive its encoded inputs $g_\mathrm{A}(\alpha_n),g_\mathrm{B}(\alpha_n)\in\mathbb{F}^{m\times(m/k)}$. AlphaTensor represents MM as a 3-way tensor that maps the bilinear form vectorized inputs to the output via 
\begin{align}
    \operatorname{vec}(\mathrm{C}) = \sum_{r=1}^{T} u_r^\top \operatorname{vec}(\mathrm{A})v_r^\top\operatorname{vec}(\mathrm{B})w_r,
\label{eq:tensorform}
\end{align}
where $T\in\mathbb{N}$ is the tensor rank (number of scalar multiplications), and $(u_r,v_r,w_r)$ are the low-rank decompositions RL agent searches for to minimize $T$.

We adapt this formulation to our finite-field setting.
Each agent holds masked inputs and executes only the local operations corresponding to its learned decomposition pattern $\{u_r,v_r,w_r\}_{r=1}^{T_l}$ with a learned rank $T_l$. Formally, for agent $n$ we define
\begin{align}
    &M(\alpha_n) \nonumber\\
    &\triangleq \sum_{r=1}^{T_l}\Big(\langle u_r, \operatorname{vec}(g_\mathrm{A}(\alpha_n))\rangle \cdot\langle v_r, \operatorname{vec}(g_\mathrm{B}(\alpha_n))\rangle\Big) \cdot \operatorname{mat}(w_r),
\label{eq:agent-learned}
\end{align}
where $\operatorname{mat}(\cdot)$ reshapes the vector $w_r$ back into a $(m/k)\times(m/k)$ matrix block. Each agent thus performs fewer scalar multiplications, replacing the standard $(m/k)^3$ cost with $T_l\ll (m/k)^3$ while keeping the same algebraic structure required by PSMM.

\subsection{Algorithmic Integration}
\label{subsec:algorithm}
The LA-PSMM protocol preserves all privacy guarantees of PSMM with the only modification in the \textit{local multiplication} step as a learned low-rank decomposition instead of a dense matrix product. We present the procedure below.

\begin{proposition}[Agent bound under tensorized local multiplications]
\label{prop:agent-bound}
Let $\mathrm{\mathrm{A,B}}\in \mathbb{F}^{m\times m}$ be $k$ partitioned respectively as~\eqref{eq:theo}
with $\mathrm{A}_i,\mathrm{B}_j \in \mathbb{F}^{m\times (m/k)}$.
Consider the encoded polynomials~\eqref{eq:theo-polyA}-\eqref{eq:theo-polyB}
as in Theorem~\ref{theorem:info-bound}, and let $\alpha_1,\dots,\alpha_N \in \mathbb{F}$ be distinct public evaluation points.
Suppose that each agent $n\in\mathbb{N}_1^N$ does not form the dense product $g_\mathrm{A}(\alpha_n)^\top g_\mathrm{B}(\alpha_n)$ directly, but instead evaluates the same bilinear map via a rank-$T$ tensor decomposition, i.e.
\begin{align}
    &M(\alpha_n) \nonumber\\
    &= \sum_{r=1}^T
    \big\langle u_r, \operatorname{vec}(g_\mathrm{A}(\alpha_n)) \big\rangle
    \big\langle v_r, \operatorname{vec}(g_\mathrm{B}(\alpha_n)) \big\rangle
    \operatorname{vec}(w_r),
\label{eq:tensor-local}
\end{align}
for some fixed triplets $(u_r,v_r,w_r)$ over $\mathbb{F}$ that realize the usual MM on $(m/k)\times(m/k)$ blocks (e.g., Strassen or an AlphaTensor-empowered decomposition). Then, the resulting global product polynomial
$M(x) \triangleq g_\mathrm{A}(x)^\top g_\mathrm{B}(x)$ has the same set of nonzero exponents as in Theorem~\ref{theorem:info-bound}, and all desired block multiplications $\mathrm{A}_i^\top \mathrm{B}_j$ appear at exponents $i-1 + k(j-1)$, while every other exponent contains at least one random mask. Consequently, the number of agents $N$ required to interpolate $M(x)$ (entrywise) is unchanged and still satisfies \eqref{eq:info-bound}.
\end{proposition}

\begin{proof}
\label{prop-proof:agent-bound}
The proof follows the structure of Theorem~\ref{theorem:info-bound}, but with the role of the local bilinear realization being explicit.
\end{proof}

The implementation of Proposition~\ref{prop:agent-bound} is illustrated in Algorithm~\ref{algo:la-psmm}, and we provide a lemma to show that the learning integration privacy-preserving property is still maintained.

\begin{lemma}[Operator choice does not affect privacy or recovery threshold]
\label{lemma:Operator-invariance}
For any available operator $M\in\mathcal{M}$ whose local computation implements the exact bilinear map
$M(\alpha)=g_\mathrm{A}(\alpha)^\top g_\mathrm{B}(\alpha)$ for all $\alpha\in\mathbb{F}$, the LA-PSMM protocol in Algorithm~\ref{algo:la-psmm} achieves the same
(i) perfect privacy against any $t-1$ colluding agents and
(ii) agent threshold $N^\star(k,t)$
as Theorem~\ref{theorem:info-bound}.
\end{lemma}

\begin{proof}
\label{lemma-proof:Operator-invariance}
Consider LA-PSMM with the matrices $\mathrm{A,B}\in\mathbb{F}^{m\times m}$ and any colluding agents set $\mathcal{T}\subseteq\mathbb{N}_1^N$ with $|\mathcal{T}|\leq t-1$. Agent $n$  receives $\big(g_\mathrm{A}(\alpha_n), g_\mathrm{B}(\alpha_n)\big)$ and then calculate the output $\hat{M}(\alpha_n)$. By assumption on the operator $M\in{\mathcal{M}}$, $\forall\alpha \in \mathbb{F}$, we have $\hat{M}(\alpha)=M(\alpha)=g_{A}(\alpha)^\top g_{B}(\alpha)$, thus, $\hat{M}(\alpha_n)$ is a deterministic function of the shares $\big(g_\mathrm{A}(\alpha_n), g_\mathrm{B}(\alpha_n) \big)$.

\noindent \textbf{From privacy perspective}: The joint distributed shares $\{g_\mathrm{A}(\alpha_{n'}), g_\mathrm{B}(\alpha_{n'}):{n'}\in \mathcal{T}\}$ are independent of $\mathrm{A,B}$ by the same masking arguments as in Theorem~\ref{theorem:info-bound} (the random masks provide prefect $t-1$ privacy). Since $\{\hat{M}(\alpha_{n'}):{n'}\in \mathcal{T}\}$ is obtained by deterministic post-processing of these shares, it cannot reveal additional information, hence, LA-PSMM is still perfectly private against any $t-1$ colluding agents.

\noindent \textbf{From recovery threshold perspective}: The controller receives $\{\hat{M}(\alpha_n):n\in\mathbb{N}_1^N\}=\{g_\mathrm{A}(\alpha_n)^\top g_\mathrm{B}(\alpha_n):n\in\mathbb{N}_1^N\}$, i.e., the same evaluation values of $M(x)=g_\mathrm{A}(x)^\top g_\mathrm{B}(x)$ as Theorem~\ref{theorem:info-bound}. Therefore, the same interpolation succeeds from the same number of agents, namely $N^\star(k,t)$, which still preserves the reconstruction. This completes the proof.
\end{proof}

\begin{algorithm}[t]
\caption{Learning-Augmented Perfectly Secure Matrix Multiplication (LA-PSMM)}
\label{algo:la-psmm}
\begin{algorithmic}[1]
\Require Matrices $\mathrm{A},\mathrm{B}\in \mathbb{F}^{m \times m}$, parameters $(k,t,N)$, learned decomposition $\{(u_r,v_r,w_r)\}_{r=1}^{T_l}$\;
\State Partition $\mathrm{A},\mathrm{B}$ into $k$ column blocks as in \eqref{eq:theo} \;
\State Perform Polynomial encodings $g_\mathrm{A},g_\mathrm{B}$ as in \eqref{eq:theo-polyA},~\eqref{eq:theo-polyB} with random masks $\mathrm{R}_\ell^{\mathrm{A}},\mathrm{R}_\ell^{\mathrm{B}}$ \;
\State Choose distinct public points $\alpha_1,\dots,\alpha_N \in \mathbb{F}$ \;
\State Send to each agent $n$ the evaluations $g_\mathrm{A}(\alpha_n),g_\mathrm{B}(\alpha_n)$ \;
\For{$n = 1 : N $}
\State Agent $n$ receives $g_\mathrm{A}(\alpha_n),g_\mathrm{B}(\alpha_n)$
\State \textbf{Initialize} $M(\alpha_n)\leftarrow\mathbf{0}^{(m/k)\times(m/k)}$
\For{$r = 1 : T_l $}
    \State $a_r\leftarrow\big\langle u_r, \operatorname{vec}(g_\mathrm{A}(\alpha_n)) \big\rangle$ 
    \State $b_r\leftarrow\big\langle v_r, \operatorname{vec}(g_\mathrm{B}(\alpha_n)) \big\rangle$ 
    \State $M(\alpha_n)\leftarrow M(\alpha_n)+(a_r\cdot b_r)\operatorname{mat}(w_r)$ 
\EndFor
\State Agent $n$ sends $M(\alpha_n)$ to the controller\;
\EndFor
\State Controller collects $\{M(\alpha_n):n\in\mathbb{N}_1^N\}$ \;
\State Controller solves the blockwise Vandermonde system induced by $\{\alpha_n\}$ to interpolate $M(x)=g_\mathrm{A}(x)^\top g_\mathrm{B}(x)$ \;
\State Extract coefficients $M_{i-1+k(j-1)} = \mathrm{A}_i^\top\mathrm{B}_j$ $\forall i,j\in\mathbb{N}_1^k$ \;
\State Stack $\{\mathrm{A}_i^\top\mathrm{B}_j\}$ to form $\mathrm{A}^\top\mathrm{B}$ \;

\Ensure Multiplication $\mathrm{A}^\top \mathrm{B}$ is reconstructed at the controller 
\end{algorithmic}
\end{algorithm}

\section{Numerical Results}
\label{sec:numerical}
In this section, we provide numerical simulations to support our theoretical framework.

\subsection{Comparisons with BGW-style job-splitting}
\label{App:compare-BGW}
We compare our PSMM scheme with a BGW-style job-splitting baseline under a per-agent storage limit $1/k$ and passive collusion of up to $t-1$ agents, with matrices $\mathrm{A, B}$ over a large prime field and sized $m=1024$. Our PSMM shows gains in terms of the number of agents required versus the storage split, and the per-agent communication near threshold. 
\subsubsection{Agents Required vs. Storage Split}
Fig.~\ref{fig:thresholds} shows the \textit{agents required} $N$ as a function of $k$ for several $t$ values. For our scheme, we use the proven threshold
\begin{align}
    N_{\text{ours}} (k,t)=\min\{\,2k^2+2t-3, k^2+kt+t-2\,\},
\label{eq:ourN}
\end{align}
while for BGW job-splitting we use $N_{\text{BGW}}(k,t)=k^2(2t-1)$.

\textbf{Insights-I:} 
(i) Across all $k$ and $t$, our $N_{\text{ours}}$ is \textit{strictly smaller} than $N_{\text{BGW}}$, with the gap widening as either $k$ or $t$ increases.  
(ii) The curvature shows two regimes in~(\ref{eq:ourN}): when $k<t$, $N$ follows $2k^2+2t-3$; when $k\ge t$, it follows $k^2+kt+t-2$, yielding even lower thresholds for larger $k$.  
(iii) Practically, this means the same privacy level ($t$) can be achieved with far fewer agents by using the structured \textit{zero-gap} polynomial sharing rather than naive job-splitting.

\begin{figure}[ht]
  \centering
  \includegraphics[width=0.8\linewidth]{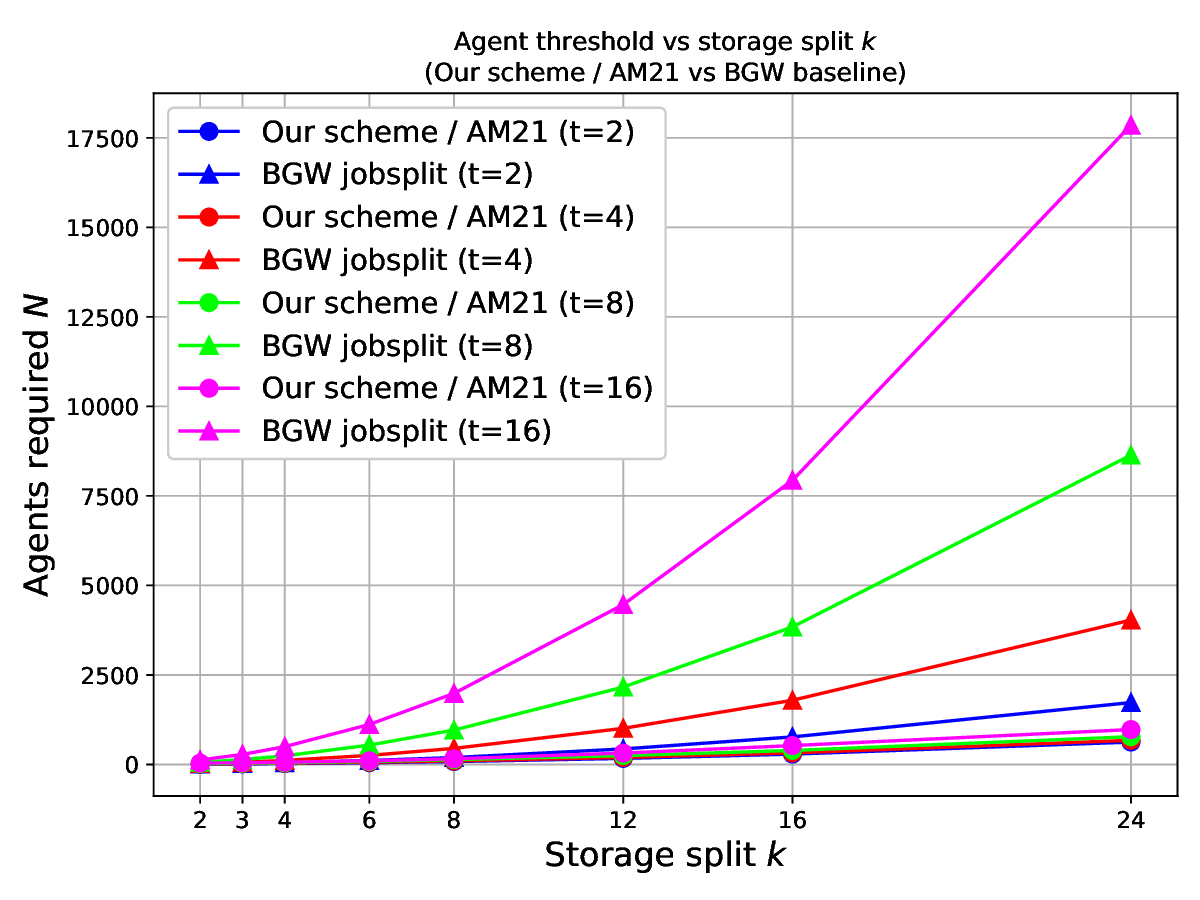}
  \caption{Agents required $N$ vs. storage split $k$ for multiple $t$.}
  \label{fig:thresholds}
\end{figure}

\begin{figure}[htp]
  \centering
  \includegraphics[width=0.8\linewidth]{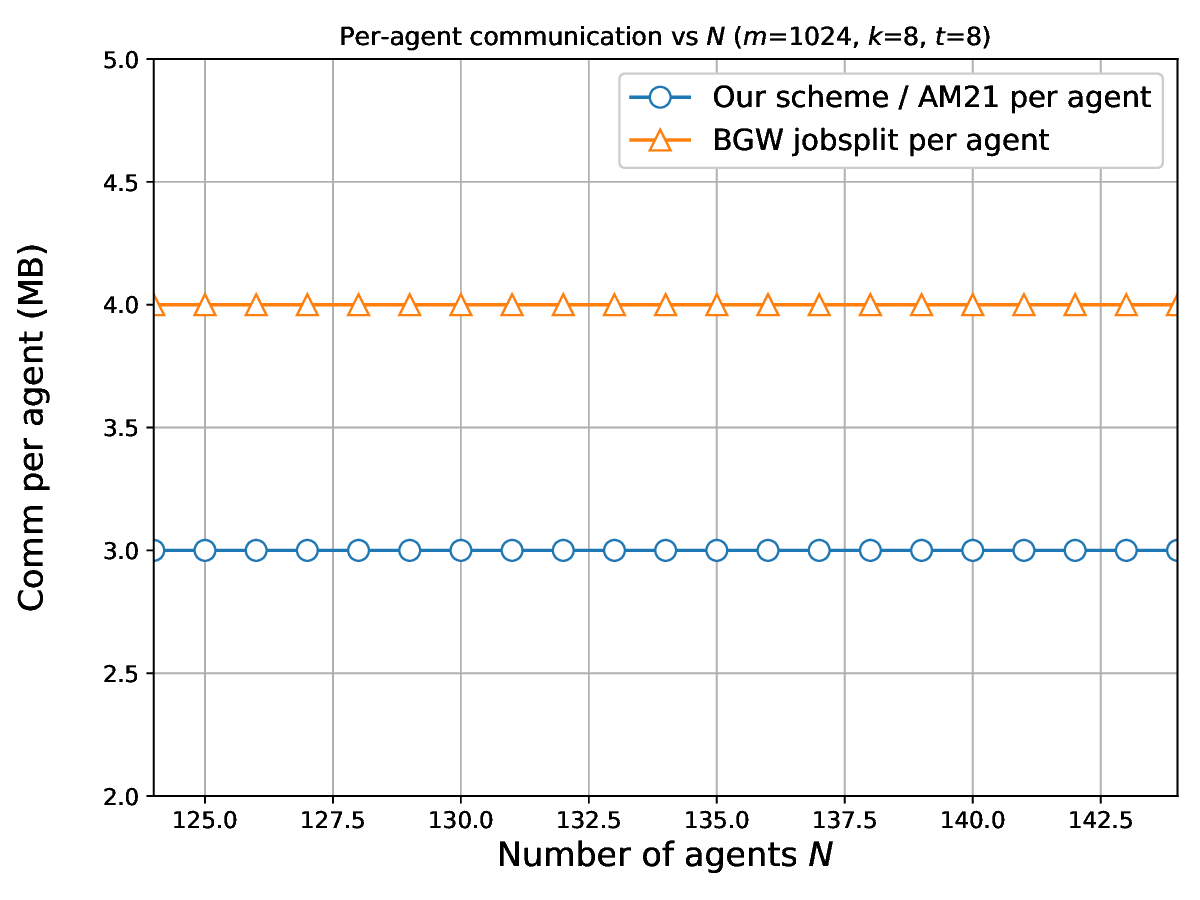}
  \caption{Per-agent communication vs.\ number of agents $N$ for $(m,k,t)=(1024,8,8)$. Our scheme requires fewer bytes per agent than BGW and also operates at a smaller $N$.}
  \label{fig:comm}
\end{figure}

\subsubsection{Per-agent Communication Near Threshold}
Fig.~\ref{fig:comm} depicts \textit{per-agent communication} (MB) versus $N$ for a representative configuration $(m,k,t)=(1024,8,8)$ around the threshold $N_{\text{ours}}(8,8)$. We count one upload of $(g_\mathrm{A}(\alpha_n),g_\mathrm{B}(\alpha_n))$ and one download of $M(\alpha_n)$ for our scheme. The BGW baseline requires more sub-shares/messages, modeled here as a larger constant factor.

\textbf{Insights-II:} (i) Communication per agent is essentially independent of $N$ for both schemes since each agent transmits a fixed-size payload; however, our scheme consistently requires fewer shares than BGW, leading to lower per-agent communication. (ii) Moreover, since our scheme operates with a smaller number of agents (Fig.~\ref{fig:thresholds}), the \textit{total} system communication $N\times\text{(per-agent)}$ of our scheme is correspondingly smaller. (iii) When combined with the fact that each agent only conducts a single local multiplication, this yields superior scalability in both communication and computation budgets.

\textbf{Takeaways}: 
The experiments corroborate the theory, namely, structured polynomial sharing with intentional exponent “gaps” compresses the set of nonzero product coefficients, lowering the interpolation/recovery threshold from $k^2(2t-1)$ (BGW job-splitting) to \eqref{eq:ourN}. In practice, this translates into \textit{fewer agents}, \textit{less communication}, and \textit{comparable per-agent compute}, making our PSMM substantially more efficient at the same information-theoretic privacy level.

\subsection{Comparison of computational time complexity}
We show that LA-PSMM has lower computational complexity than PSMM. In PSMM, each agent performs a coded dense multiplication between an $(s\times m)$ and an $(m\times s)$ matrix with computational complexity $\mathcal{O}(m^3/k^2)$ per agent and $\mathcal{O}(N m^3/k^2)$ total. In LA-PSMM, the dense product is replaced by a learned low-rank expansion of length $T_l$, where each agent evaluates two linear forms on vectors of length $ms=m^2/k$ and accumulates $T_l$ rank-1 terms into an $s\times s$ output with computational complexity $\mathcal{O}(T_l m^2/k)$ per agent and $\mathcal{O}(N T_l m^2/k)$ total, yielding a local computational reduction whenever an improvement when $T_l \ll m/k$. Encoding remains $\mathcal{O}(N(k+t)m^2/k)$, and decoding via entrywise interpolation is $\mathcal{O}((m/k)^2N^2)$ with a naive Vandermonde solver (or near-linear in $N$ per entry with fast interpolation).
\begin{figure}[h]
\centering
\includegraphics[width=0.8\linewidth]{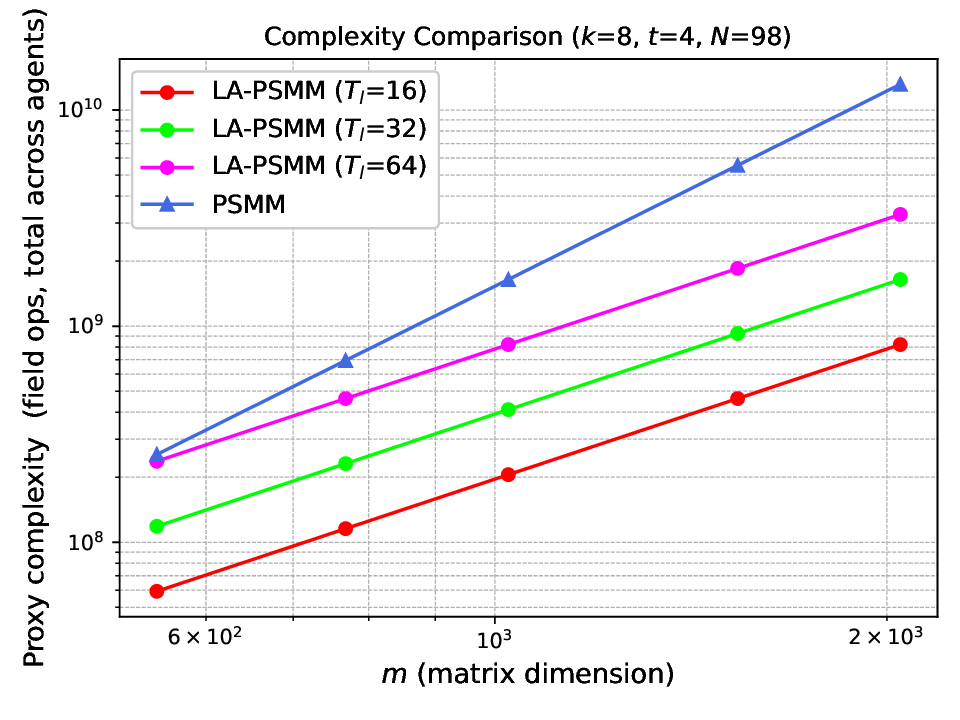}
    \caption{The agent computation complexity comparison of PSMM vs. LA-PSMM with different learned rank $T_l$.}
    \vspace{-0.1cm}
    \label{fig:compare-complexity}
\end{figure}

\begin{figure}[h]
    \centering
    \includegraphics[width=0.8\linewidth]{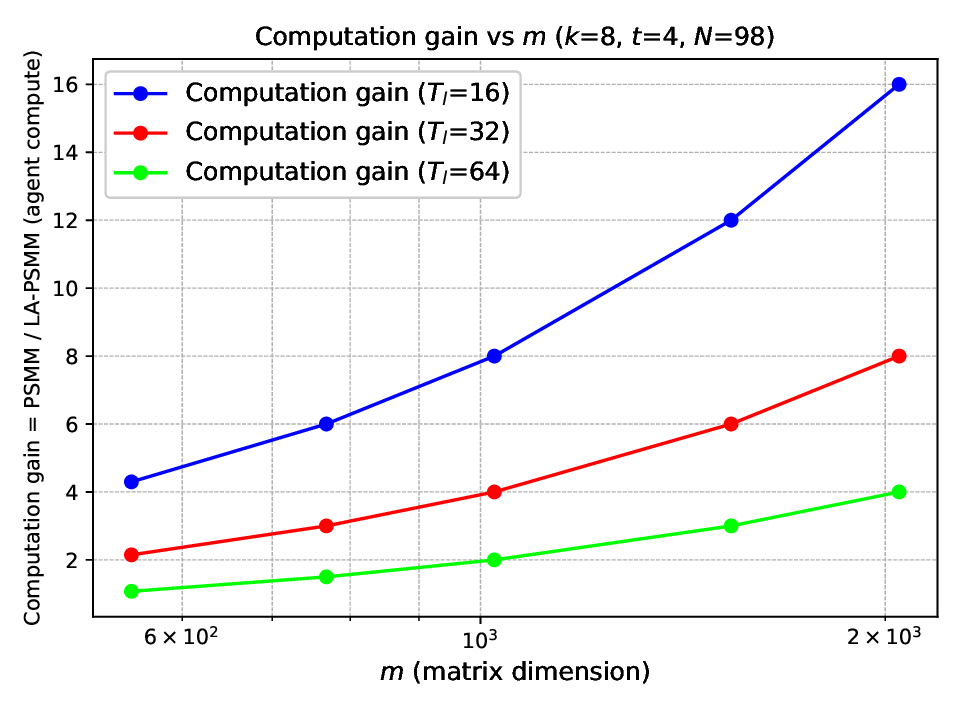}
    \caption{The agent computation gain of LA-PSMM with different learned rank $T_l$.}
    \vspace{-0.1cm}
    \label{fig:compare-gain}
\end{figure}

Fig.~\ref{fig:compare-complexity} depicts total computational complexity across agents for $k=8$, $t=4$, and $N=98$. LA-PSMM consistently requires less computation than PSMM, and the gap grows with dimension $m$ (up to $80\%$), indicating that learning-augmented local multiplication increasingly amortizes cost as problem size increases. Fig.~\ref{fig:compare-gain} further shows the computation gain, defined as the ratio of PSMM to LA-PSMM local cost, which rises monotonically with $m$, reflecting increasing benefit from the learning-based tensor decomposition. Overall, LA-PSMM delivers substantial, scalable savings in computation while preserving PSMM’s security guarantees.


\clearpage
\newpage

\IEEEtriggeratref{17}

\bibliographystyle{IEEEtran}
\bibliography{literature_conf}
\clearpage

\appendix
\subsection{Proof of Proposition~\ref{prop:beaver}}
\label{app:prop-proof:beaver}
Since $\mathrm{A} = \mathrm{R}_1 + \mathrm{D}$ and $\mathrm{B} = \mathrm{R}_2 + \mathrm{E}$ according to \eqref{eq:beaver-open}, multiplying them gives
$\mathrm{A}\mathrm{B} = (\mathrm{R}_1 + \mathrm{D})(\mathrm{R}_2 + \mathrm{E}) = \mathrm{R}_1\mathrm{R}_2 + \mathrm{D}\mathrm{R}_2 + \mathrm{R}_1\mathrm{E} + \mathrm{D}\mathrm{E}$ in \eqref{eq:beaver-compute}.
Since $\mathrm{R}_3 = \mathrm{R}_1\mathrm{R}_2$, the latter equals to $\mathrm{R}_3 + \mathrm{D}\mathrm{R}_2 + \mathrm{R}_1\mathrm{E} + \mathrm{D}\mathrm{E}$,
opening $\mathrm{D,E}$ reveals only uniformly masked differences. Hence, the private inputs remain hidden.
\subsection{Proof of Theorem~\ref{theorem:info-bound}}

\label{app:theorem-proof:info-bound}
Since each block $\mathrm{A}_i,\mathrm{B}_j$ is with shape $m\times (m/k)$, $g_\mathrm{A}(\alpha_n)$ and $g_\mathrm{B}(\alpha_n)$ are with $m\times (m/k)$, and each agent’s local multiplication $M(\alpha_n) = g_\mathrm{A}(\alpha_n)^\top g_\mathrm{B}(\alpha_n)$ is with $(m/k)\times (m/k)$.
For coefficient placement in $M(x)$, we define it as the summation form $M(x)=g_\mathrm{A}(x)^\top g_\mathrm{B}(x)=\sum_{\nu} \mathrm{M}_\nu x^\nu,~\nu\in\mathbb{N}_0$. Then from \eqref{eq:theo-polyA}-\eqref{eq:theo-polyB}, we have the transpose
\begin{align}
   g_\mathrm{A}(x)^\top=\sum_{i=1}^{k} \mathrm{A}_i^\top x^{i-1} + \sum_{\ell=1}^{t-1} \big(\mathrm{R}^{(\mathrm{A})}_\ell\big)^\top x^{k^2+\ell-1}.
\end{align}
Multiplying the structured non-random parts yields the $k^2$ target block multiplications at exponents
\begin{align}
    \mathcal{K}_1 = \big\{(i-1) + k(j-1):\ i,j\in\mathbb{N}^k\big\} = \{0,1,\ldots,k^2-1\}.
\label{eq:K1}
\end{align}
Thus for each $(i,j)$, the coefficient $\mathrm{M}_{i-1+k(j-1)}$ contains exactly $\mathrm{A}_i^\top\mathrm{B}_j$. Cross-terms that involve one non-random part and one random mask compose
\begin{align}
    \mathcal{K}_2 = \{k^2,\,k^2+1,\ldots,k^2+k+t-3\},
\label{eq:K2}
\end{align}
from the below multiplication
\begin{align*}
    \Big(\textstyle\sum_i \mathrm{A}_i^\top x^{i-1}\Big)\Big(\textstyle\sum_\ell \mathrm{R}^{(\mathrm{B})}_\ell x^{k^2+\ell-1}\Big).   
\end{align*} 
and similarly for other indices as 
\begin{align}
    \mathcal{K}_3 = \{k^2 + ik + j:\ i\in\mathbb{N}_0^{k-1},~j\in\mathbb{N}_0^{t-2}\},
\label{eq:K3}
\end{align}
we derive the degrees from the following set
\begin{align*}
    \Big(\textstyle\sum_\ell (\mathrm{R}^{(\mathrm{A})}_\ell)^\top x^{k^2+\ell-1}\Big)\Big(\textstyle\sum_j \mathrm{B}_j x^{k(j-1)}\Big).
\end{align*}
Finally, masks times masks create a high-degree band as: 
\begin{align}
    \mathcal{K}_4 = \{2k^2,\,2k^2+1,\ldots,2k^2+2t-4\}.
\label{eq:S4}
\end{align}
Consequently, every index in $\mathcal{K}_2\cup \mathcal{K}_3\cup \mathcal{K}_4$ carries at least one random block and thus acts only as \textit{noise} for privacy, while $\mathcal{K}_1$ holds all desired $\mathrm{A}_i^\top \mathrm{B}_j$.

We next focus on the minimum number of agents needed given the yielding $\mathcal{K}_1,\mathcal{K}_2,\mathcal{K}_3,\mathcal{K}_4$ and the sparsity of $M(x)$.

\noindent\textbf{Counting nonzero coefficients}:
The degree pattern implies that the set of \textit{distinct} exponents with nonzero coefficients in $M(x)$ is contained in $\mathcal{K}_1\cup \mathcal{K}_2\cup \mathcal{K}_3\cup \mathcal{K}_4$. A careful overlap analysis standard in polynomial-sharing arguments shows that the total number of occupied exponents is
\begin{align}
    |\operatorname{supp}(M)| \leq \min\{2k^2+2t-3,  k^2+kt+t-2\}, 
\end{align}
where $\operatorname{supp}(M)$ denotes the operator counting the non-zero elements in $M(x)$.

\noindent\textbf{Intuition}: $\mathcal{K}_1$ contributes $k^2$ consecutive indices, $\mathcal{K}_2$ contributes a contiguous run of length $(k+t-2)$ starting at $k^2$; $\mathcal{K}_3$ contributes $k(t-1)$ indices spaced by $k$ (some may fall inside $\mathcal{K}_2$ depending on $k$ vs.\ $t$), and $\mathcal{K}_4$ contributes a final run of length $(2t-3)$ starting at $2k^2$. When $k<t$, the overlaps are smaller, yielding the $2k^2+2t-3$ regime; when $k\ge t$, more $\mathcal{K}_3$ indices merge into the $\mathcal{K}_2$ band, which offers the tighter $k^2+kt+t-2$ count.

\noindent\textbf{Interpolation requirement and choice of $N$}: Let $K = |\operatorname{supp}(M)|$ denote the number of unknown coefficients in $M$. Suppose we observe $N\geq K$ distinct evaluations $\{M(\alpha_n)\}: n\in\mathbb{N}_1^{N}$ at pairwise distinct points $\alpha_n\in\mathbb{F}$. Each evaluation yields a linear equation in the unknown coefficients $\{\mathrm{M}_\nu:\nu\in\operatorname{supp}(M)\}$, and collecting all $N$ evaluations gives a blockwise (entrywise) Vandermonde linear system.

\noindent For a Zariski-open set of evaluation points $(\alpha_1,\ldots,\alpha_N)$, this Vandermonde system has full rank~\cite{weil1946foundations}, and hence admits a unique solution. When the underlying field $\mathbb{F}$ is sufficiently large, the Schwartz–Zippel lemma implies that choosing the $\alpha_n$ independently and uniformly at random ensures invertibility with probability at least $1-\mathcal{O}(1/|\mathbb{F}|)$. Consequently, $N=M$ evaluations are sufficient to recover all coefficients $\mathrm{M}_\nu$ exactly.

\noindent In particular, this allows recovery of the $k^2$ target blocks in $\mathcal{K}_1$, namely $\mathrm{A}_i^\top\mathrm{B}_j$, for $i,j\in\mathbb{N}^k$. Stacking these blocks yields the desired multiplication $\mathrm{A}^\top\mathrm{B}$.

\noindent \textbf{Privacy and correctness}:
By construction, any set of $t-1$ agents sees only evaluations of low-degree random polynomials in the masked positions. By the masking lemma (uniformity of evaluations at $t-1$ points), these are independent of the inputs, so nothing about $\mathrm{\mathrm{A,B}}$ is leaked beyond the final output. Correctness follows because interpolation recovers the exact coefficients in $\mathcal{K}_1$. If Beaver triples are used for local multiplications, Proposition~\ref{prop:beaver} ensures each local multiplication and the recombination are algebraically correct over $\mathbb{F}$. Combining the counting bound for $K$ with the interpolation argument proves the claimed threshold in \eqref{eq:info-bound} that gives
\begin{align*}
  N \leq \min\{2k^2+2t-3,~k^2+kt+t-2\}. 
\end{align*}
Now for the privacy part, we can continue as follows:

\noindent \textbf{Mask polynomials}:
Recall that we split $\mathrm{A,B}$ into $k$ partitioned blocks of size $m\times(m/k)$, and encode them in \eqref{eq:theo-polyA}, \eqref{eq:theo-polyB} with random masks $\mathrm{R}^{(\mathrm{A})}_\ell,\mathrm{R}^{(\mathrm{B})}_\ell\in\mathbb{F}^{m\times(m/k)}$ as in \eqref{eq:theo}.
Each agent $n$ then receives the evaluations $g_\mathrm{A}(\alpha_n),g_\mathrm{B}(\alpha_n)$ and computes $g(\alpha_n)= g_\mathrm{A}(\alpha_n)^\top g_\mathrm{B}(\alpha_n)$.
By defining the masking components as
\begin{align}
    r_\mathrm{A}(x) &= \sum_{\ell=1}^{t-1} \mathrm{R}^{(\mathrm{A})}_\ell\,x^{\ell-1}, 
\label{eq:rA}\\
    r_\mathrm{B}(x) &= \sum_{\ell=1}^{t-1} \mathrm{R}^{(\mathrm{B})}_\ell\,x^{\ell-1}, 
\label{eq:rB}
\end{align}
the multiplicated polynomials can then be written as structured parts plus $x^{k^2}r_\mathrm{A}(x)$ and $x^{k^2}r_\mathrm{B}(x)$ as shown in \eqref{eq:theo-polyA} and \eqref{eq:theo-polyB}.

\noindent{\bf Privacy argument}:
For any coalition $\mathcal{S}$ of at most $t-1$ agents, the joint distribution
\begin{align}
\label{eq-joiny}
    \big\{g_\mathrm{A}(\alpha_n),g_\mathrm{B}(\alpha_n): n\in\mathcal{S}\big\},
\end{align}
is uniform and independent of $\mathrm{A, B}$, since by~(\ref{eq:rA})-(\ref{eq:rB}) the coalition only sees at most $t-1$ evaluations of degree-$\leq t-2$ random polynomials, which by the Shamir secret sharing property are information-theoretically independent of the secrets. 
Therefore, the scheme inherits privacy directly from the \textit{secret sharing model}. {This completes the proof.}


\subsection{Strassen Matrix Multiplication}
\label{app:str-mat-mult}
For completeness, we summarize the classical Strassen algorithm~\cite{strassen:1969}, which forms the basis for our low-rank learned block multiplication in Sec.~\ref{sec:learning-mul}. Strassen’s method replaces the standard $2\times 2$ block MM, which requires $8$ scalar multiplications, with only $7$ bilinear forms.

\noindent{\bf Formulation}:
Let we have two matrices $\mathrm{A}$ and $\mathrm{B}$ as follows
\begin{align}
    \mathrm{A} = 
    \begin{bmatrix}
    a_{11} & a_{12}\\
    a_{21} & a_{22}
    \label{eq:app-matrix}
    \end{bmatrix},
    \qquad
    \mathrm{B} = 
    \begin{bmatrix}
    b_{11} & b_{12}\\
    b_{21} & b_{22}
\end{bmatrix}, 
\end{align}
where each $a_{ij}, b_{ij}\in\mathbb{F}$ are scalar or block entries (depending on the recursive level). The product $\mathrm{C}=\mathrm{A}\mathrm{B}$ is given by
\begin{align}
    \mathrm{C} =
    \begin{bmatrix}
    c_{11} & c_{12}\\
    c_{21} & c_{22}
\label{eq:app-matrix-result}
\end{bmatrix}.
\end{align}

Strassen observed that $\mathrm{C}$ can be computed using the following seven
intermediate products:
\begin{align}
    {\begin{cases}
    F_1 &= (a_{11} + a_{22})(b_{11} + b_{22}), \nonumber\\
    F_2 &= (a_{21} + a_{22})\, b_{11}, \nonumber\\
    F_3 &= a_{11} \,  (b_{12} - b_{22}), \nonumber\\
    F_4 &= a_{22} \,  (b_{21} - b_{11}), \nonumber\\
    F_5 &= (a_{11} + a_{12})\,  b_{22}, \nonumber\\
    F_6 &= (a_{21} - a_{11})\,  (b_{11} + b_{12}), \nonumber\\
    F_7 &= (a_{12} - a_{22})\,  (b_{21} + b_{22}).\nonumber
    \end{cases}}
\end{align}

The final result is reconstructed as
\begin{align}
    c_{11} &= F_1 + F_4 - F_5 + F_7, \nonumber\\
    c_{12} &= F_3 + F_5, \nonumber\\
    c_{21} &= F_2 + F_4, \nonumber\\
    c_{22} &= F_1 - F_2 + F_3 + F_6. 
\label{eq:strassen-recomb}
\end{align}

\textbf{Remarks}:
\begin{itemize}
    \item The algorithm replaces $8$ scalar (or block) multiplications with $7$, at the expense of additional additions/subtractions.
    \item When applied recursively, the asymptotic complexity reduces from
$\mathcal{O}(n^3)$ to $\mathcal{O}(n^{\log_2 7}) \approx \mathcal{O}(n^{2.807})$.
\item  In the finite-field setting $\mathbb{F}_p$, all operations are
performed modulo $p$, which preserves the correctness of \eqref{eq:strassen-recomb}.
\item   In our LA-PSMM implementation, we employ the blockwise form of Strassen's relations to construct the low-rank bilinear maps
$(u_r,v_r,w_r)$ for $r\in\mathbb{N}_1^7$ in \eqref{eq:agent-learned}.
\end{itemize}

\paragraph*{Tensor Form (Low-Rank Interpretation)}
From \eqref{eq:strassen-recomb}, we can equivalently represent the multiplication as a
rank-$7$ decomposition of the $2\times2$ matrix multiplication tensor:
\begin{align}
    \operatorname{vec}(\mathrm{C}) = \sum_{r=1}^{7}
    \langle u_r, \operatorname{vec}(\mathrm{A})\rangle
    \langle v_r, \operatorname{vec}(\mathrm{B})\rangle
    \operatorname{vec}(w_r),
\end{align}
where the vectors $\{u_r,v_r,w_r\}$ encode the additive structure of
$\{F_r:r\in\mathbb{N}_1^N\}$ above.
This tensor view underlies both Strassen’s method and
modern learning-based low-rank decompositions such as
AlphaTensor~\cite{fawzi:2022matrixRL}.

\subsection{Proof of Proposition~\ref{prop:adv-dof}}
\label{App:Proof-prop:adv-dof}
We outline the argument in the same coefficient-support framework used in Theorem~\ref{theorem:info-bound}.

\noindent{\bf Step 1} (same encoding and privacy):
We use the same polynomials $g_\mathrm{A},g_\mathrm{B}$ as in \eqref{eq:theo-polyA}--\eqref{eq:theo-polyB}, with $t-1$ random mask blocks placed in the high-degree tail.
Therefore, for any coalition $\mathcal{S}$ with $|\mathcal{S}|\leq t-1$, the view
$\{g_\mathrm{A}(\alpha_n),g_\mathrm{B}(\alpha_n):n\in\mathcal{S}\}$ is statistically independent of $(\mathrm{A,B})$ by Lemma~\ref{lemma:masking}.
Hence, perfect privacy remains unchanged and holds under the same $t$.

\noindent{\bf Step 2} (support and interpolation in the worst case):
Let $M(x)=g_\mathrm{A}^\top(x) g_\mathrm{B}(x)=\sum_{\nu} \mathrm{M}_\nu x^\nu$.
Theorem~\ref{theorem:info-bound} identifies a set $\mathcal{K}_1$ in \eqref{eq:K1} of $k^2$ exponents corresponding to the $k^2$ target blocks
$\mathrm{M}_{i-1+k(j-1)}=\mathrm{Z}_{i,j}$ as defined in \eqref{eq:taget-block}, while all remaining exponents in $\mathcal{K}_2\cup \mathcal{K}_3\cup \mathcal{K}_4$ contain at least one mask.
In the worst case, recovering all $k^2$ unknown blocks plus the masked coefficients requires
$N \ge |\operatorname{supp}(M)| \leq N^\star(k,t)$ evaluations.

\noindent{\bf Step 3}  (structural reduction of unknowns):
Under the DOF assumption in \eqref{eq:dof-model}, the set of $k^2$ target blocks
$\{\mathrm{Z}_{i,j}\}$ defined in \eqref{eq:taget-blocks} is not an arbitrary element of
$\big(\mathbb{F}^{(m/k)\times(m/k)}\big)^{k^2}$, but lies in a known $s$-dimensional linear subspace.
Equivalently, there exist $s$ unknown \textit{latent} blocks (the $\mathcal{L}_\ell(\mathrm{A,B})$ values) from which all $\mathrm{Z}_{i,j}$ are deterministically recoverable via known linear combinations.
Thus, it suffices to recover these $s$ latent blocks (and not all $k^2$ coefficients in $\mathcal{K}_1$) to reconstruct $\mathrm{A}^\top\mathrm{B}$.

\noindent{\bf Step 4} (decoding only $s$ effective coefficients):
Since the coefficient placement is fixed and known, we may replace the full coefficient vector
$\{\mathrm{M}_\nu:\nu\in \operatorname{supp}(M)\}$ by a reduced unknown vector consisting of:
(i) the $s$ latent blocks determining all $\mathrm{Z}_{i,j}$, and
(ii) the same masked coefficients required for privacy (which are independent noise terms).
The interpolation system remains linear in these unknowns; therefore, the required number of evaluations scales with the number of unknown blocks being solved for.
Replacing $k^2$ independent target blocks by $s$ latent blocks yields the count in \eqref{eq:adv-newN} by repeating the same overlap/support counting argument from Theorem~\ref{theorem:info-bound}, with $k^2$ replaced by $s$ in the structured band.

\noindent{\bf Step 5} (strict advantage when $s<k^2$):
When $s<k^2$, the reduced threshold in \eqref{eq:adv-newN} is strictly smaller than $N^\star(k,t)$, hence for a fixed $N$ we can tolerate a larger $t-1$ collusion level while maintaining perfect privacy as shown in \eqref{eq:tolerable-collusion}.
This completes the proof.

\end{document}